\newif\ifdraft \draftfalse
\newif\iffull \fulltrue

\makeatletter \@input{tex.flags} \makeatother
\documentclass[11pt]{article}
\usepackage{fullpage}

\usepackage{algorithm}
\usepackage[noend]{algorithmic}
\usepackage{hyperref}

\usepackage{amsmath, amssymb, amsthm}
\usepackage{mathtools}
\usepackage{verbatim}
\usepackage{color}
\usepackage{xcolor}

\usepackage{multicol}

\definecolor{DarkGreen}{rgb}{0.1,0.5,0.1}
\definecolor{DarkRed}{rgb}{0.5,0.1,0.1}
\definecolor{DarkBlue}{rgb}{0.1,0.1,0.5}
\usepackage{hyperref}
\hypersetup{
    unicode=false,          
    pdftoolbar=true,        
    pdfmenubar=true,        
    pdffitwindow=false,      
    pdftitle={},    
    pdfauthor={}
    pdfsubject={},   
    pdfnewwindow=true,      
    pdfkeywords={keywords}, 
    colorlinks=true,       
    linkcolor=DarkRed,          
    citecolor=DarkGreen,        
    filecolor=DarkRed,      
    urlcolor=DarkBlue,          
}
\usepackage{xspace}
\usepackage{cleveref}
\usepackage{thmtools, thm-restate}

\usepackage[numbers]{natbib}

\newcommand{\ar}[1]{\ifdraft \textcolor{brown}{[AR: #1]}\fi}
\newcommand{\sw}[1]{\ifdraft \textcolor{blue}{[SW: #1]}\fi}
\newcommand{\mk}[1]{\ifdraft \textcolor{red}{[MK: #1]}\fi}

\newcommand\RR{\mathbb{R}}
\newcommand\cA{\mathcal{A}}

\newcommand\cF{\mathcal{F}}

\newcommand\cM{\mathcal{M}}
\newcommand\cP{\mathcal{P}}
\newcommand\cG{\mathcal{G}}

\newcommand\cO{\mathcal{O}}
\newcommand\cR{\mathcal{R}}
\newcommand\cX{\mathcal{X}}

\newcommand\cN{\mathcal{N}}
\newcommand\cT{\mathcal{T}}
\newcommand\cS{\mathbf{s}}

\newcommand{\cI}{\mathcal{I}}

\newcommand{\searc}{\mbox{{\sf SearchCom}\xspace}}
\newcommand{\bfs}{\mbox{{\sf SFS}\xspace}}
\newcommand{\SV}{\mbox{{\sf SV}\xspace}}
\newcommand{\mc}{\mbox{{\sf Target}\xspace}}
\newcommand{\pmc}{\mbox{{\sf PTarget}\xspace}}
\newcommand{\infect}{\mbox{{\sf infect}\xspace}}
\newcommand{\CN}{\mbox{{\sf CN}\xspace}}

\renewcommand{\tilde}{\widetilde}

\newcommand{\eps}{\varepsilon}
\def\epsilon{\varepsilon}

\newcommand{\IC}{\mathrm{IC}}
\DeclareMathOperator{\Lap}{Lap}

\renewcommand{\hat}{\widehat}

\DeclareMathOperator*{\argmax}{\mathrm{argmax}}

\newcommand{\INDSTATE}[1][1]{\STATE\hspace{#1\algorithmicindent}}

\newtheorem*{theorem*}{Theorem}

\declaretheorem[
  name=Theorem,
  refname={theorem, theorems},
  Refname={Theorem, Theorems}]{theorem}
\declaretheorem[
  name=Lemma,
  refname={lemma, lemmas},
  Refname={Lemma, Lemmas}]{lemma}

\declaretheorem[
  name=Remark,
  refname={remark, remarks},
  Refname={Remark, Remarks}]{remark}

\declaretheorem[
  name=Definition,
  refname={definition, definitions},
  Refname={Definition, Definitions}]{definition}

\title{Privacy for the Protected (Only)}

\author{
Michael Kearns\thanks{Department of Computer and Information Sciences, University of Pennsylvania. Email:
\href{mailto:mkearns@cis.upenn.edu}{mkearns@cis.upenn.edu}}
\and
Aaron Roth\thanks{Department of Computer and Information Sciences, University of Pennsylvania. Email: \href{mailto:aaroth@cis.upenn.edu}{aaroth@cis.upenn.edu}}
\and
Zhiwei Steven Wu\thanks{Department of Computer and Information Sciences, University of Pennsylvania. Email: \href{mailto:wuzhiwei@cis.upenn.edu}{wuzhiwei@cis.upenn.edu}}
\and
Grigory Yaroslavtsev\thanks{Department of Computer and Information Sciences, University of Pennsylvania. Email: \href{mailto:grigoryy@seas.upenn.edu}{grigoryy@seas.upenn.edu}}
}

\begin{document}

\pagenumbering{gobble}
\maketitle

\begin{abstract}
Motivated by tensions between data privacy for individual citizens,
and societal priorities such as counterterrorism and the containment of infectious disease,
we introduce a computational model that distinguishes between parties for whom
privacy is explicitly protected, and those for whom it is not (the {\em targeted\/} subpopulation). 
The goal is the development
of algorithms that can effectively identify and take action upon members of the targeted subpopulation
in a way that minimally compromises the privacy of the protected,
while simultaneously limiting the
expense of distinguishing members of the two groups via costly mechanisms such as surveillance, background checks, or
medical testing.
Within this framework, we provide provably privacy-preserving algorithms for targeted search in social
networks. These algorithms are natural variants of common graph search methods, and ensure privacy
for the protected by the careful injection of noise in the prioritization of potential targets.
We validate the utility of our algorithms with extensive computational experiments on two large-scale
social network datasets.
\end{abstract}

\vfill
\pagebreak
\pagenumbering{arabic}
\section{Introduction}

The tension between useful or essential gathering and analysis of
data about citizens, and the privacy rights of those citizens, is at an
historical peak. Perhaps the most striking and controversial recent example is
the revelation that U.S. intelligence agencies systemically engage in
``bulk collection'' of civilian ``metadata'' detailing telephonic and
other types of communication and activities, with the alleged purpose
of monitoring and thwarting terrorist activity~\cite{Snowden}.
Other compelling examples abound, including in medicine (patient privacy
vs. preventing epidemics), marketing (consumer privacy vs. targeted
advertising), and many other domains.

Debates about (and models for) data privacy often have an ``all or
nothing'' flavor: privacy guarantees are either provided to every
member of a population, or else
privacy is deemed to be a failure.  This dichotomy is only appropriate
if all members of the population have an equal right to, or demand
for, privacy.  Few would argue that actual terrorists should have such
rights, which leads to difficult questions about the balance between
protecting the rights of ordinary citizens, and using all available
means to prevent terrorism.\footnote{A recent National Academies
study~\citep{NAS} reached the conclusion that there are not (yet)
technological alternatives to bulk collection and analysis of civilian
metadata, in the sense that such data is essential in current
counterterrorism practices.}
A major question is whether and when the former
should be sacrificed in service of the latter.  Similarly, in the
medical domain, epidemics (such as the recent international outbreak
of Ebola~\cite{Ebola})
have raised serious debate about the clear public interest in
controlling contagion versus
the privacy
rights of the infected and those that care for them.

The model and results in this paper represent a step towards
explicit acknowledgments of such trade-offs, and algorithmic methods
for their management. The scenarios sketched above can be broadly
modeled by a population divided into two types. There is a {\em
  protected subpopulation\/} that enjoys (either by law, policy, or
choice) certain privacy guarantees. For instance, in the examples
above, these protected individuals might be non-terrorists, or
uninfected citizens (and perhaps informants and health care professionals).  They are
to be contrasted with the ``unprotected'' or {\em targeted\/} subpopulation,
which does not share those privacy assurances. A key assumption
of the model we will introduce is that the protected or targeted status
of individual subjects is not known, but can be discovered by (possibly costly)
measures, such as surveillance or background investigations (in the case of terrorism)
or medical tests (in the case of disease).
Our overarching goal is to allow parties such as intelligence or medical agencies
to identify and take
appropriate actions on the targeted subpopulation, while also
providing privacy assurances for the protected individuals who are not
the specific targets of such efforts --- all while limiting the cost and extent
of the background investigations needed.

As a concrete example of the issues we are concerned with, consider
the problem of using social network data (for example, telephone
calls, emails and text messages between individuals) to search for
candidate terrorists. One natural and broad approach would be to
employ common graph search methods: beginning from known terrorist
``seed'' vertices in the network, neighboring vertices are
investigated, in an attempt to ``grow'' the known subnetwork of
targets.\footnote{This general practice is sometimes referred to as
  ``contact chaining'': {\em ``Communications metadata, domestic and
    foreign, is used to develop contact chains by starting with a
    target and using metadata records to indicate who has communicated
    with the target (1 hop), who has in turn communicated with those
    people (2 hops), and so on. Studying contact chains can help
    identify members of a network of people who may be working
    together; if one is known or suspected to be a terrorist, it
    becomes important to inspect others with whom that individual is
    in contact who may be members of a terrorist network.''}  Section
  3.1 of~\cite{NAS}.}  A major concern is that such search methods
will inevitably encounter protected citizens, and that even taking
action against only discovered targeted individuals may compromise the
privacy of the protected.



In order to rigorously study the trade-offs between privacy and
societal interests discussed above, our work introduces a
formal model for privacy of network data that provides provable assurances {\em only\/} to
the protected subpopulation, and gives algorithms that allow effective
investigation of the targeted population.
These algorithms are
deliberately ``noisy'' and are privacy-preserving versions of the
widely used graph search methods mentioned above, and as such
represent only mild (but important) departures from commonly used
approaches.
At the highest level, one can think of our algorithms as outputting
a list of targeted individuals discovered in the network for which
any subsequent action (e.g. publication in a most-wanted list, further
surveillance or arrest in the case of terrorism, or
medical treatment or quarantine in the case of epidemics)
will not compromise the privacy of the protected.



The key elements of our model include the following:
\begin{enumerate}
 \item Network data collected over a population of individuals and consisting of pairwise
	contacts (physical, social, electronic, financial, etc.). The
	contacts or links of each individual comprise the private
	data they desire to protect. We assume a third party
	(such as an intelligence agency or medical organization) has
	direct access to this network data, and would like to discover
	and act upon targeted individuals.
\item For each individual, an immutable {\em status bit\/} that
	determines their membership status in the targeted subpopulation
	(such as terrorism or infection). These status bits can be discovered 
	by the third party, but only at some nontrivial cost
	(such as further surveillance or medical testing), and thus there is
	a \textit{budget}
	limiting the number of status bits that an algorithm can reveal.
	One might assume or hope that in practice, this budget is sufficient to
	investigate a number of individuals that is of the order of the targeted
	subpopulation size, but considerably less than that needed to investigate
	every member of the general population.
\item A mathematically rigorous notion of individual data privacy (based on the widely studied \emph{differential privacy}~\citep{DMNS06}) that provides guarantees of privacy for the network data of only the protected individuals,
while allowing the discovery 
of targeted individuals. Informally, this notion guarantees that
compared to a counterfactual world in which any protected individual \emph{arbitrarily} changed any part of their data, or even removed themselves entirely from the computation, their risk (measured with respect to the probability of \emph{arbitrary} events) has not substantially increased.
\end{enumerate}

Our main results are:
\begin{enumerate}
\item The introduction of a broad class of graph search algorithms
  designed to find and identify targeted
  individuals.
  This class of algorithms is based on a general notion of a {\em statistic
    of proximity\/} --- a network-based measure of how ``close'' a
  given individual $v$ is to a certain set of individuals $S$. For
  instance, one such closeness measure is the number of short paths in
  the network from $v$ to members of $S$.  Our (necessarily
  randomized) algorithms add noise to such statistics in order to
  prioritize which status bits to query (and thus how to spend the
  budget).
\item A theoretical result providing a quantitative privacy guarantee for this class of algorithms,
where the level of privacy depends on a measure of the \textit{sensitivity} of the statistic of proximity
to small changes in the network. 
\item Extensive computational experiments in which we demonstrate the effectiveness of
	our privacy-preserving algorithms on real social network
	data. 
	These experiments demonstrate that in addition to the privacy guarantees,
	our algorithms are also
	useful, in the sense that they find almost as many members of
	the targeted subpopulation as their non-private
	counterparts. The experiments allow us to quantify the loss
	in effectiveness incurred by the gain in privacy.
\end{enumerate}

We note that although our class of network search algorithms is relatively broad, it necessarily excludes some natural and commonly used algorithms. This is by design, since some algorithms are clearly in conflict with the kind of privacy we wish to protect for protected individuals.


To our knowledge, our formal framework is the first to introduce
explicit protected and targeted subpopulations with qualitatively
differing privacy rights,\footnote{This is in contrast to the
  quantitative distinction proposed by Dwork and McSherry \cite{DM10},
  which still does not allow for the explicit discovery of
  targeted individuals.} and our algorithms the first to provide
mathematically rigorous privacy guarantees for the protected while
still allowing effective discovery of the targeted. More generally, we
believe our work is a first step towards richer privacy models
that acknowledge and manage the tensions between different levels of
privacy guarantees to different subgroups.

\section{Preliminaries}




Consider a social network in which the individuals are
partitioned into a~\emph{targeted subpopulation} $\cT$ and
a~\emph{protected subpopulation} $\cP$. Individuals correspond to the vertices $V$ in the network, and the private data of each individual $v$ is the set of edges incident to $v$. Each individual also has an immutable
{\em status bit\/} which specifies to which subpopulation the individual belongs.  We assume that the value of this bit is not easily
observed, but can be discovered through (possibly costly)
investigation. Our goal is to develop search algorithms to identify
members of the targeted subpopulation, while preserving the privacy of the edge set of the
protected population.

Any practical algorithm must operate under an investigation {\em budget\/}, which limits the number of status bits
that are examined. Our goal is a total number of status bit examinations that is
on the order of the size of the targeted subpopulation $\cT$, which may be much smaller than the size of the protected population $\cP$.
This is the source of the tension we study --- because the budget is limited, it is necessary to exploit the private edge set to 
guide our search (i.e. we cannot simply investigate the entire population), 
but we wish to do so in a way that does not reveal much about the edges incident to any specific individual.

The privacy guarantee we provide is a variant of~\emph{differential
  privacy}, an algorithmic definition of data privacy. It formalizes
the requirement that arbitrary changes to a single individual's
private data should not significantly affect the output distribution
of the data analysis procedure, and so guarantees that the analysis
leaks little information about the private data of any single
individual. We first introduce the definition of differential privacy
specialized for the network setting.\footnote{This definition is also
  known as vertex differential privacy, and is the strongest version
  of differential privacy for networks that is used in the literature
  (cp. edge differential privacy). It is a variant of a slightly more
  general original definition of differential
  privacy~\cite{DMNS06}. Vertex differential privacy was first defined
  by Hay et al.~\cite{HLMJ09} and later studied by Kasiviswanathan et
  al.~\cite{KNRS13} and Blocki et al.~\cite{BBDS13}.} We treat
networks as a collection of vertices representing individuals, each
represented as a list of its edges (which form the private data of
each vertex).  For a network $G$ and a vertex $v$, let $D_v(G)$ be the
set of edges incident to the vertex $v$ in $G$.  Let $\mathcal G_n$ be
the family of all $n$-vertex networks.

\begin{definition}[Vertex Differential Privacy~\citep{DMNS06,HLMJ09}]
  The networks $G, G'$ in $\mathcal G_n$ are~\emph{neighboring} if one
  can be obtained from the other by an (arbitrary) rewiring of the
  edges incident to a single vertex $v$ --- i.e. if for some vertex
  $v$, $D_u(G) \setminus \{(u,v)\}= D_u(G') \setminus \{(u,v)\}$ for
  all $u \neq v$.  An algorithm $\cA \colon \mathcal G_n \rightarrow
  \cO$ satisfies $\eps$-differential privacy if for every event $S
  \subseteq \cO$ and all neighboring networks $G, G'\in \mathcal G_n$,
 \[ \Pr[\cA(G) \in S] \leq e^\eps \Pr[\cA(G') \in S].  \]
\end{definition}

Differential privacy is an extremely strong guarantee --- it has many interpretations (see discussion in e.g. \cite{DR14}), but most straightforwardly, it promises the following: simultaneously for every individual $i$, and simultaneously for any event $S$ that they might be concerned about, 
event $S$ is \emph{almost} no more likely to occur given that individual $i$'s data is used in the computation, 
compared to if it were replaced by an arbitrarily different entry.
Here, ``almost no more likely'' means that the probability that the bad event $S$ occurs has increased by a multiplicative factor of at most $e^{\epsilon}$, which we term the \emph{risk multiplier}. As the privacy parameter $\epsilon$ approaches $0$, the value of the risk multiplier approaches $1$, meaning that agent $i$'s data has no effect at all on the probability of a bad outcome.
The smaller the risk multiplier, the more meaningful the privacy guarantee. It will be easier for us to reason directly about the privacy parameter $\epsilon$ in our analyses, but semantically it is the risk multiplier $e^{\epsilon}$ that measures the quality of the privacy guarantee, and it is this quantity that we report in our experiments.


\mk{This paragraph is dense and confusing, and has lots of run-on
  sentences.  It is a crucial paragraph for the reader's
  understanding. Aaron can you try to
  streamline/clarify?}\ar{Attempted rewording:} Differential privacy
promises the same protections for~\emph{every} individual in a
network, which is incompatible with our setting. We want to be able to
identify members of the targeted population, and to do so, we want to
be able to make arbitrary inferences from their network data.
Nevertheless, we want to give strong privacy guarantees to members of
the protected subpopulation. This motivates our variant of
differential privacy, which redefines the neighboring
relation.\footnote{This is in contrast to other kinds of relaxations
  of differential privacy, which relax the worst-case assumptions on
  the prior beliefs of an attacker as in Bassily et al.~\cite{BGKS13},
  or the worst-case collusion assumptions on collections of data
  analysts as in Kearns et al.~\cite{KPRU14}.}
In contrast to the definition of neighbors given above, we now say that 
two networks are neighbors if one can be obtained from the other by arbitrarily re-wiring the edges incident to \emph{a single member of the protected population}. Crucially, networks are not considered to be neighbors if they differ in either:
\begin{enumerate}
\item The way in which they partition vertices between the protected and targeted populations $\cP$ and $\cT$, or
\item In \emph{any} edges that connect pairs of vertices $u, v \in \cT$ that are both members of the targeted population.
\end{enumerate}
 
What this means is that we are offering no
guarantees about what an observer can learn about either the status of
an individual (protected vs. targeted), or the set of edges incident to 
targeted individuals. However, we are still promising that no observer can learn much about the set of edges incident to 
any member of the protected subpopulation. This naturally leads us to the following definition:

\begin{definition}[Protected Differential Privacy]
	Two networks $G, G'$ in $\mathcal G_n$ are \emph{neighboring}
        if they:
	\begin{enumerate}
	 \item Share the same partition into $\cP$  and $\cT$, and
	  \item $G$ can be obtained from $G'$ by rewiring the set of edges incident to a single vertex $v \in \cP$.
	 \end{enumerate}
	 An algorithm $\cA \colon \cG_n \rightarrow \cO$ satisfies
	$\eps$-\emph{protected differential privacy} if for any two
	neighboring networks $G, G'\in \cG_n$, and for any event $S\subseteq
	\cO$:
	\[
	\Pr[\cA(G) \in S] \leq e^\eps\Pr[\cA(G') \in S].
	\]
\end{definition}

Formally, our network analysis algorithms take as input a network and a method by which they may query whether vertices $v$ are members of the protected population $\cP$ or not. The class of algorithms we consider are network search algorithms --- they aim to identify some subset of the targeted population. Our privacy guarantees are oblivious as to what action is taken on the identified members (for example, in a medical application they might be quarantined, in a security application they might be arrested, etc.), but we assume that whatever action is taken might be observable. Hence, without loss of generality we can abstract away the action taken and simply view the output of the mechanism to be an ordered list of targeted individuals.


\begin{figure}
  \centering
  \includegraphics[scale=0.35]{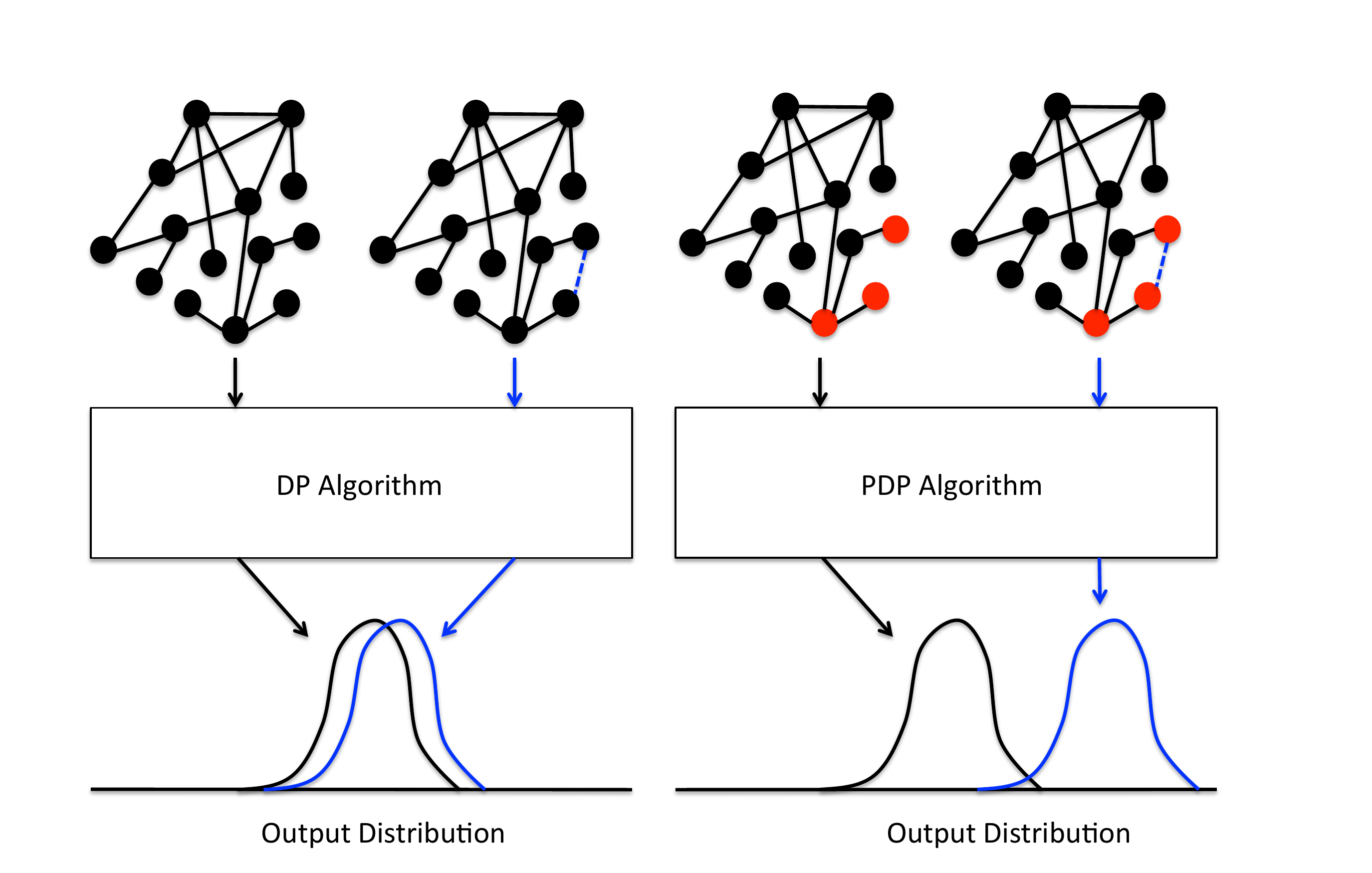}
  \caption{
	Informal illustration of standard Differential Privacy (DP) versus
	Protected Differential Privacy (PDP).
        For algorithms satisfying standard DP (left), the addition of a single edge (dashed blue)
	can alter the output distribution by only a small amount.
	PDP is similar, except we introduce a targeted subpopulation (highlighted in red). If the added edge is between
	two targeted individuals,
	the output distribution may change
	arbitrarily, reflecting the fact that the targeted parties may not enjoy privacy protection.
	The formal definitions are stronger, in that privacy for protected individuals must be preserved
	even if any number of edges to them is added or deleted.}
  \label{fig:PDP}
\end{figure}

\section{Algorithmic Framework}

The key element in our algorithmic framework is the notion
of a~\emph{Statistic of Proximity} (SoP), a network-based measure of how
close an individual is to another set of individuals in a network. Formally, an SoP
is a function $f$ that takes as input a graph $G$, a vertex $v$ and a set of
targeted vertices $S \subseteq \cT$, and outputs a numeric value $f(G, v, S)$. Examples
of such functions include the number of common neighbors between $v$ and
the vertices in $S$, and
the number of short paths from $v$ to $S$.

Algorithms in our framework rely on the SoP to prioritize which status bits
to examine. Since the value of the SoP depends on the protected vertex's
private data, we perturb the values of the SoP by adding
noise with scale proportional to its~\emph{sensitivity}, which
captures the magnitude by which a single protected vertex can affect the
SoP of some \emph{targeted} vertex. Let $G \sim G'$ denote two neighboring
networks in $\cG$. The sensitivity of the SoP $f$ is defined as:
\[
\Delta(f) = \max_{G\sim G', t \in \cT, S\subseteq \cT} \left |f(G , t, S) -
f(G', t, S)\right|.
\]

Crucially, note that in this definition --- in contrast to what is
typically required in standard differential privacy --- we are only
concerned with the degree to which a protected individual can affect
the SoP of a targeted individual.

We next describe the non-private version of our targeted search algorithm $\mc(k, f)$.
For any fixed SoP $f$, $\mc$ proceeds
in $k$ rounds, each corresponding to the identification of a new connected component in the subgraph induced by $\cT$. The algorithm must be
started with a ``seed vertex'' --- a pre-identified member of the targeted population.
Each round of the algorithm consists of two steps:


\begin{enumerate}
\item \emph{Statistic-First Search}: Given a seed targeted vertex, the
  algorithm iteratively grows a discovered component of targeted vertices, by examining, in order of their SoP values, the vertices that neighbor the previously discovered targeted vertices. This continues
   until every neighbor of the discovered members of the
  targeted population has been examined, and all of them have been
  found to be members of the protected population. We note that this
  procedure discovers every member of the targeted population that is
  part of the same \emph{connected component} as the seed vertex, in the
  subgraph induced by only the members of the targeted population.

\item \emph{Search for a New Component}: \mk{The word ``seed''
  should be reserved for only the original targeted vertices, not new
  components found by the algorithm. We should pick something else, please fix here and globally.}\ar{Changing to ``initial vertex''}
  Following the
  completion of statistic-first search, the algorithm must find a new vertex in the targeted population to serve as an initial vertex to begin a new round of statistic-first search.  To do this, the
  algorithm computes the value of the SoP $f$ evaluated on each
  unexamined vertex, using as the input set $S$ the set of already
  discovered members of the targeted population.  It then sorts
  all of the vertices in decreasing order of their SoP value,
  and begins examining them in this order. The first vertex that
  is found to be a member of the targeted population is used as an initial
  vertex in the next iteration (taking the place of our seed vertex). We skip this search procedure in the
  last iteration.\footnote{In the Technical Appendix, we 
    present a slight variant of this procedure that allows the search
    algorithm to halt if it is unable to find any new targeted vertices
    after some number of examinations.\sw{added} }


\end{enumerate}
The algorithm outputs discovered targeted individuals as they are found, and so its output can be viewed as being an ordered list of targeted individuals.

The private version of the targeting algorithm $\pmc(k, f, \eps)$, is
a simple variant of the non-private version. The statistic-first search
stage remains unchanged, and only the search for a new component is
modified. In the private variant, when the algorithm computes the
value of the SoP $f$ on each unexamined vertex, it then perturbs each
of these values with noise sampled from the Laplace
distribution\footnote{We use $\Lap(b)$ to denote the Laplace
  distribution centered at 0 with probability density function:
  $\Pr(x) = \frac{1}{2b}\exp\left(-\frac{|x|}{b} \right)$.}
$\Lap(\triangle(f)/\eps)$ where $\eps$ is a parameter. Finally, it
examines the vertices in sorted order of their \emph{perturbed} SoP
values.

We prove the following, deferring details of the proof and the
algorithm to the 
Technical Appendix:
\begin{theorem}
\label{thm:main}
Given any $k \ge 1$ and $\eps >0$ and a fixed SoP $f$, the algorithm $\pmc(k,
f, \eps)$ recovers $k$ connected components of the subgraph induced by the targeted vertices and satisfies $((k-1)\cdot \eps)$-protected differential privacy.
\end{theorem}

\mk{This sentence is confusing; first we haven't used the term
  statistic-first search, and second, it makes it seem like the algo
  is not randomized. What point are we trying to make?}  \ar{Statistic
  first search is what ``statistic first search'' was called in an
  earlier draft -- changing it back to be consistent. This sentence is
  trying to point out to people who know a bit about DP what is
  surprising in our setting.}
There are two important things to note about this theorem.
First, we obtain a privacy guarantee despite the fact that the
statistic-first search portion of our algorithm is not randomized ---
only the search for new components employs randomness. Second, the
privacy cost of the algorithm grows only with $k$, the number of
disjoint connected components of targeted individuals (disjoint in the
subgraph defined on targeted individuals), and \emph{not} with the
total number of individuals examined, or even the total number of
targeted individuals identified. Hence, the privacy cost can be very
small on graphs in which the targeted individuals lie only in a small
number of connected components or ``cells''. Both of these features are unusual
when compared with typical guarantees that one can obtain under the
standard notion of differential privacy.

\begin{figure*}[!]
  \centering
 \includegraphics[scale=0.30]{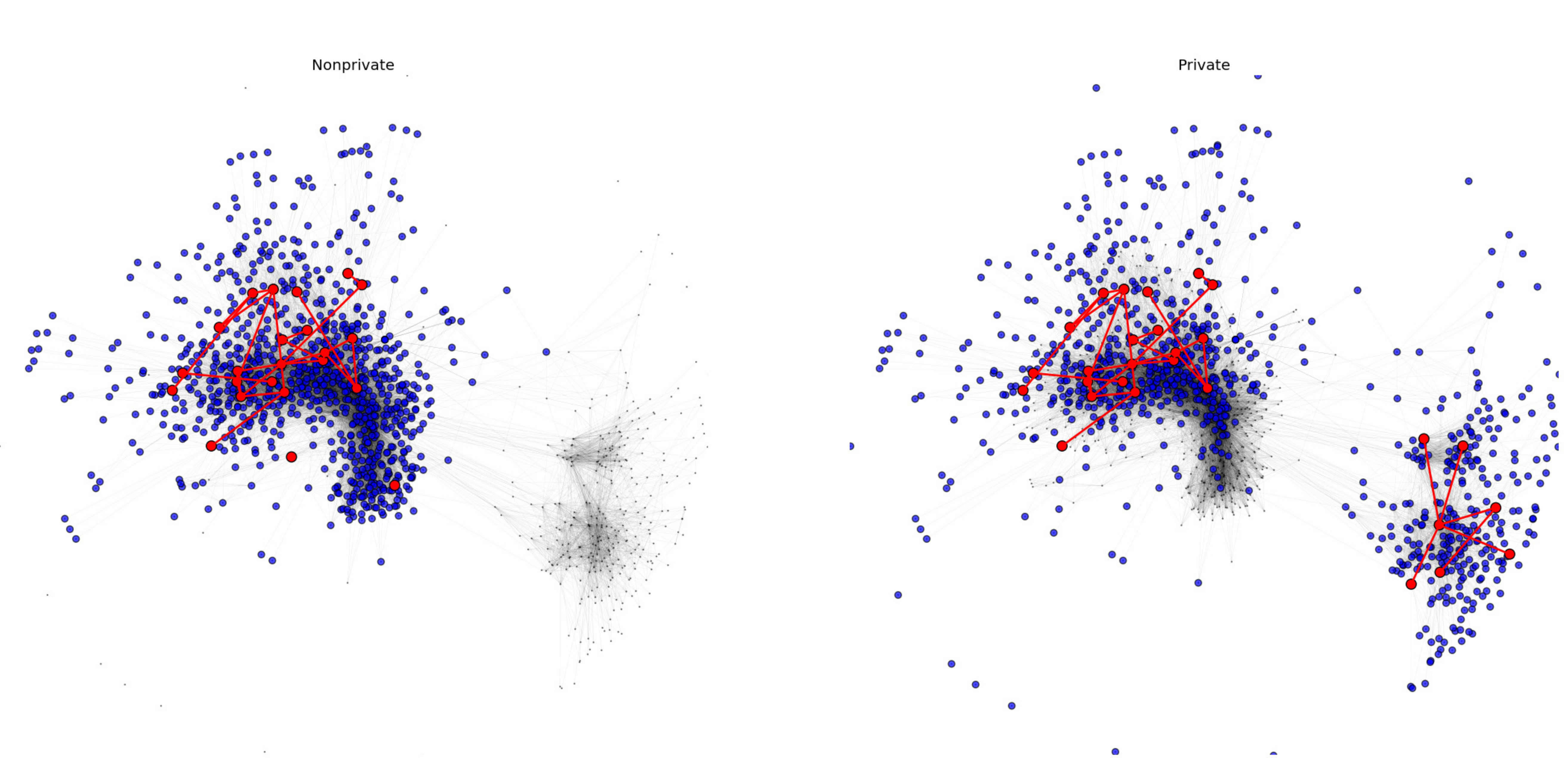}
 \caption{ \mk{This figure is orphaned somewhat far from the
      reference to it, can we try to fix?}  Visual comparison of the
    non-private algorithm $\mc$ (left panel) and the private algorithm
    $\pmc$ (right panel) on a small portion of the IMDB network (see
    Experimental Evaluation for more details). For each algorithm,
    blue indicates protected vertices that have been examined, red
    indicates targets that have been examined, and gray vertices have
    not been examined yet.  Both algorithms begin with the same
    seed target vertex, and by directed statistic-first search discover a subnetwork of
    targeted individuals (central red edges). As a consequence, many
    protected vertices are discovered and examined as well. Due to the
    added noise, $\pmc$ explores the network in a more diffuse
    fashion, which in this case permits it to find an additional
    subnetwork of targets towards the right side of the network. The
    primary purpose of the noise, however, is for the privacy of
    protected vertices.  }
  \label{fig:comparison}
\end{figure*}


\mk{Wouldn't it be better to start by describing the non-private, and
then discuss how the private algo modifies it?}\ar{Changed}
Because $\pmc$ adds randomness for privacy, it results in examining a different set of vertices as compared to
$\mc$. Figure~\ref{fig:comparison} provides a sample visualization of the
contrasting behavior of the two algorithms.
While theorems comparing the utility of $\mc$ and $\pmc$ are possible, they
require assumptions ensuring that the chosen SoP is sufficiently
``informative'', in the sense of separating the targeted from the
protected by a wide enough margin.  In particular, one needs to rule
out cases in which all unexplored targeted vertices are deemed closer
to the current set than all protected vertices, but only by an
infinitesimal amount, in which case the noise added by $\pmc$
eradicates all signal. In general such scenarios are unrealistic, so
instead of comparing utility theoretically, we now provide an
extensive empirical comparison.

\section{Experimental Evaluation}

In this section we empirically demonstrate the utility of our private
algorithm $\pmc$ by comparing its performance to its non-private
counterpart $\mc$. We report on computational experiments performed on
real social network data drawn from two sources --- the paper
coauthorship network of DBLP (``Digital Bibliography and Library
Project'')\cite{dblp}, and the co-appearance network of film actors of
IMDB (``Internet Movie Database'')\cite{imdb} --- whose macroscopic
properties are summarized in Table~\ref{tab:data}.  \mk{please add
  citations to these DBs so people could go find them?}

\begin{table}[ht]

\caption{Social network datasets used in the experiments.}
\label{tab:data} 
\centering 
\begin{tabular*}{0.9\hsize}{c c c c} 
\hline\hline 
Network & Number of vertices & Number of edges & Edge relation \\ [0.5ex] 
\hline 
DBLP          & 956,043   &   3,738,044 & scientific paper co-authorship\\
IMDB    &   235,710   & 4,587,715 & movie co-appearance\\ [1ex]
\hline 
\end{tabular*}
\end{table}

These data sources provide us with naturally occurring networks, but not a
targeted subpopulation. While one could attempt to use
communities within each network (e.g. all co-authors within a particular
scientific subtopic), our goal was to perform large-scale
experiments in which the component structure of targeted vertices (which
we shall see is the primary determinant of performance) could be more
precisely controlled.
We thus used a simple parametric stochastic
diffusion process (described in the 
Technical Appendix) to
generate the targeted subpopulation in each network.
We then evaluate our private
search algorithm $\pmc$ on these networks, and compare its performance
to the non-private variant $\mc$. For brevity we shall describe our results
only for the IMDB network; results for the DBLP network are quite similar.

In our experiments, we fix a particular SoP: the number of common
neighbors between the vertex $v$ and the subset of vertices $S$
representing the already discovered members of the targeted
population. This SoP has sensitivity $1$, and so can be used in
our algorithm while adding only a small amount of noise. In
particular, the private algorithm $\pmc$ adds noise sampled from the
Laplace distribution $\Lap(20)$ to the SoP when performing new
component search. By~\Cref{thm:main}, such an instantiation of $\pmc$
guarantees $((k-1)/20)$-protected differential privacy if it finds $k$
targeted components.

The main trade-off we explore is the number of members of the targeted
population that are discovered by the algorithms, as a function of the
number of status bits that have been investigated. 
In each of the ensuing plots, the $x$-axis measures the size of the
investigation budget consumed so far, while the $y$-axis measures the
number of targeted vertices identified for a given budget.  In each
plot, the parameters of the diffusion model described above were fixed
and used to stochastically generate targeted subpopulations of the
fixed networks given by our social network data. By varying these
parameters, we can investigate performance as a function of the
underlying component structure of the targeted subnetwork.  As we
shall see, in terms of relative performance, there are effectively
three different regimes of the diffusion model (i.e. targeted
subpopulation) parameter space. In all of them $\pmc$ compares
favorably with $\mc$, but to different extents and for different
reasons that we now discuss. We also plot the growth of the
risk multiplier for $\pmc$, which remains less than 2 in all three regimes.

On each plot, there is a single blue curve showing the performance of
the (deterministic) algorithm $\mc$, and multiple red curves showing
the performance across 200 runs of our (randomized) algorithm $\pmc$.

The first regime (Figure~\ref{fig:case2}) occurs when the largest
connected component of the targeted subnetwork is much larger than all
the other components.  In this regime, if both algorithms begin at a
seed vertex inside the largest component, there is effectively no
difference in performance, as both algorithms remain inside this
component for the duration of their budget and find identical sets of
targeted individuals.  More generally, if the algorithms begin at a
seed outside the largest component, relative performance is a ``race''
to find this component; the private algorithm lags slightly due to the
added noise, but is generally quite competitive.

The second regime (Figure~\ref{fig:case3}) occurs when the component sizes
are more evenly distributed, but there remain a few significantly
larger components. In this setting both algorithms spend more of their
budget outside the targeted subpopulation ``searching'' for these
components. Here the performance of the private algorithm lags more
significantly --- since both algorithms behave the same when inside of
a component, the smaller the components are, the more detrimental the
noise is to the private algorithm.

The third regime (Figure~\ref{fig:case4}) occurs when all the targeted
components are small, and thus both algorithms suffer accordingly,
discovering only a few targeted individuals.



\begin{figure*}[h]
   \centering
    \includegraphics[scale=.22]{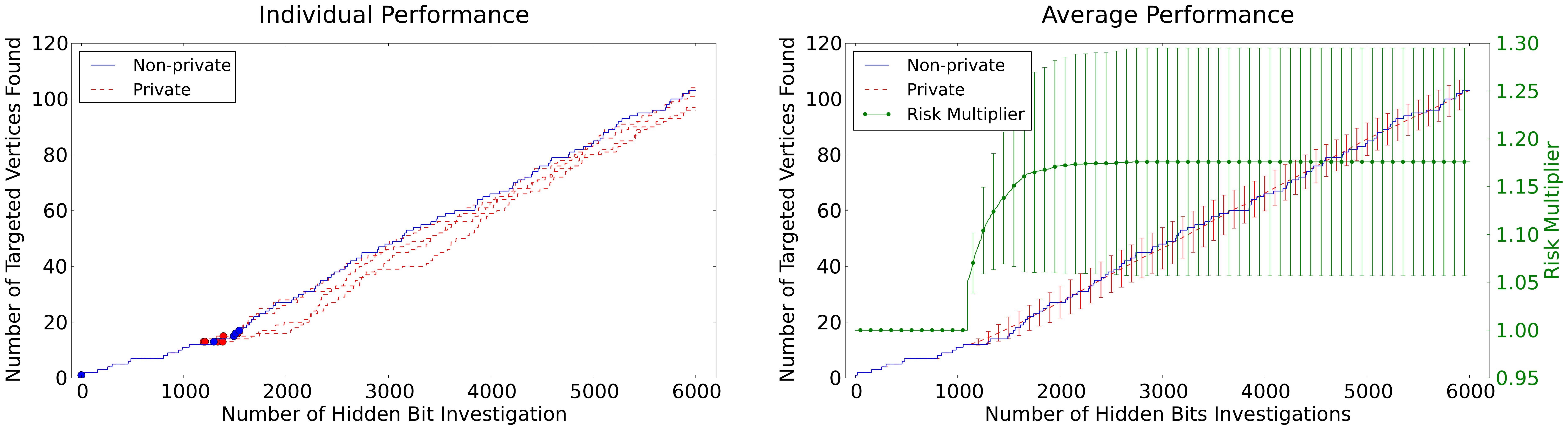}
    \caption{Performance for the cases in which there is a dominant
      component in the targeted subpopulation. In the left panel, we
      show the number of targeted vertices found as a function of the
      budget used for both the (deterministic) non-private algorithm
      $\mc$ (blue), and for several representative runs of the
      randomized private algorithm $\pmc$ (red).  Circles indicate
      points at which an algorithm has first discovered a new targeted
      component.  In the right panel, we show average performance over
      200 trials for the private algorithm with 1-standard deviation
      error bars.  We also show the private algorithm risk multiplier
      with error bars.  In this regime, after a brief initial flurry
      of small component discovery, both algorithms find the dominant
      component, so the private performance closely tracks
      non-private, and the private risk multiplier quickly levels off
      at around only 1.17.}
    \label{fig:case2}
  \end{figure*}

\begin{figure*}[!]
    \centering
    \includegraphics[scale=.22]{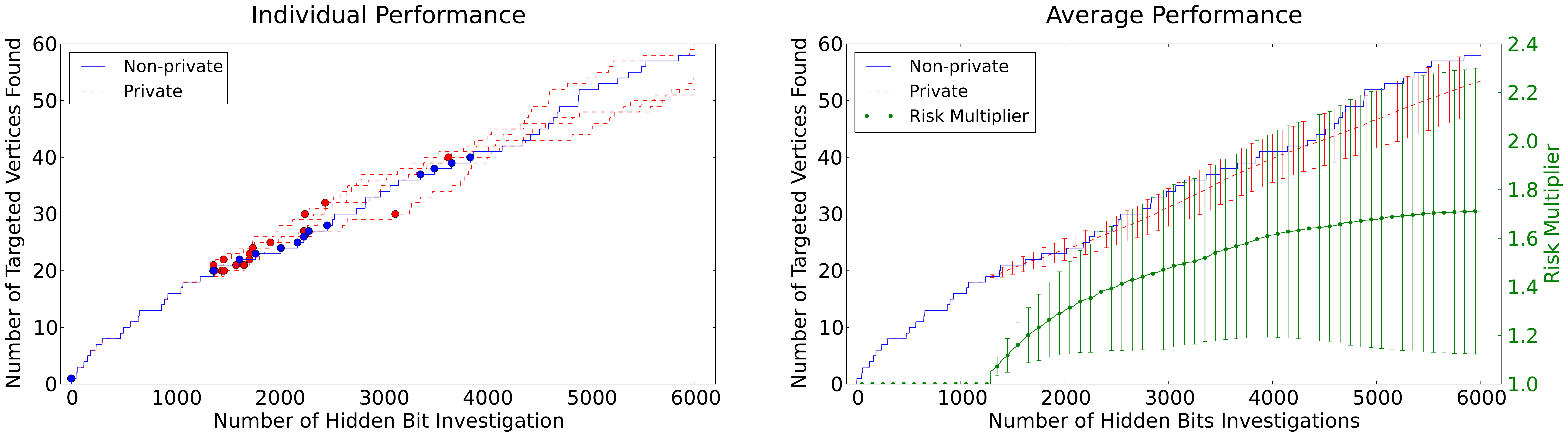}
    \caption{Same format as in Figure~\ref{fig:case2}, but now in
      a case where the component sizes are more evenly distributed,
      but still relatively large. The performance of both algorithms
      is hampered by longer time spent investigating non-targeted
      vertices (note the scale of the $y$ axis compared to
      Figure~\ref{fig:case2}). Targeted component discovery is now
      more diffuse.  The private algorithm remains competitive but
      lags slightly, and as per Theorem~\ref{thm:main} the risk
      multiplier grows as more targeted components are discovered (but
      remains modest).}
    \label{fig:case3}
  \end{figure*}
\begin{figure*}[!]
 \centering
    \includegraphics[scale=.22]{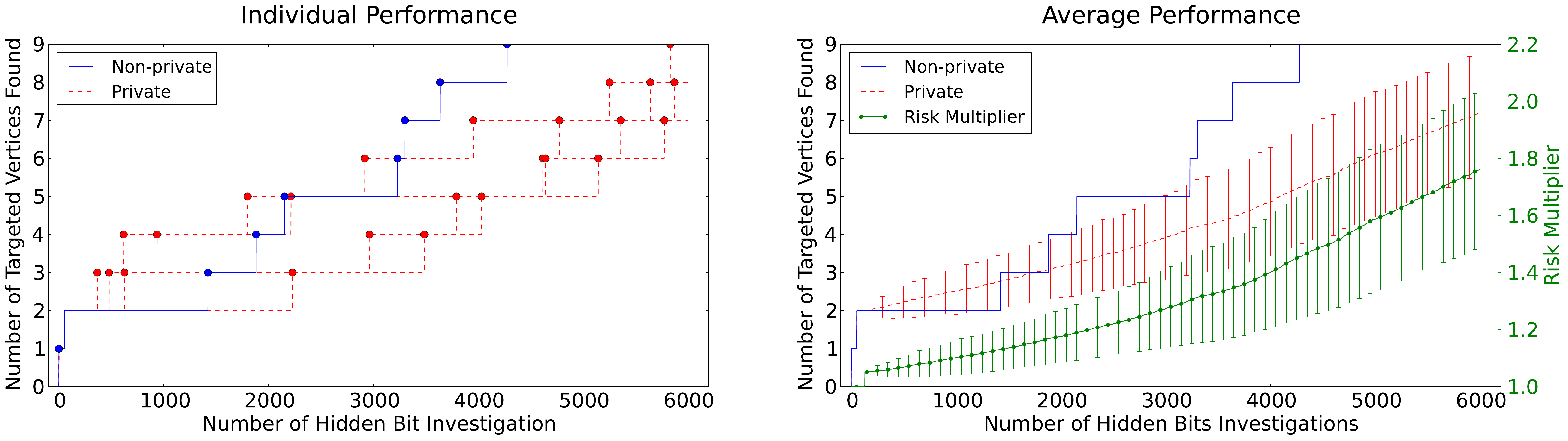}
    \caption{A case with a highly fragmented targeted
      subpopulation. Both algorithms now spend most of their budget
      investigating non-targeted vertices and suffer accordingly.}
    \label{fig:case4}
  \end{figure*}

\pagebreak

\bibliographystyle{plain}
\bibliography{Preprint-back}
\pagebreak
\appendix 
\part*{Technical Appendix}

\section{Model \& Preliminaries}
We study graph search algorithms which operate on graphs $G = (V, E)$
defined over a vertex set $V$ and edge set $E \subseteq V \times
V$. The vertex set $V$ partitioned into two fixed subsets $V
= \cT \cup \cP$, where $\cT$ represents the \emph{targeted}
subpopulation, and $\cP$ represents the \emph{protected}
subpopulation. The algorithms we consider are initially given
a single \emph{seed vertex} $s \in \cT$ (or several such vertices), and the goal of
the algorithm will be to find as many other members of the targeted
subpopulation $\cT$ as possible.

The algorithm cannot directly observe which subpopulation a particular
vertex $v \in V$ belongs to since otherwise the problem is trivial, but it has the ability to make a query on
a vertex $v \in V$ to determine its subpopulation membership. We model
this ability formally by giving the algorithm access to
an \emph{identity oracle} $\cI\colon V\to \{0, 1\}$, defined such that
$\cI(v) = 1$ if and only if $v \in \cT$. A call to this oracle is the
abstraction we use to represent the possibly costly operation
(instantiated in our example applications by e.g. surveillance, or
medical tests) which determines whether a particular member of the
population is protected or not. Because we view these operations as
expensive, we want our algorithm to operate by making as few calls to
this oracle as possible. Hence, the algorithm must use the network
data represented by the graph $G$ to \emph{guide its search} for which
vertices to query. This creates a source of privacy tension since
the edges in this network are what we view as private information.

Thus our goal is to give algorithms which discover members of the targeted
  population using the edges in the network to guide their search. We
  wish to protect the privacy of the protected individuals: we do not
  want the outcome of the search to reveal too much about the edge set
  incident to any protected individual. However, we want to exploit the
  edges incident to targeted individuals in ways that will not
  necessarily be privacy preserving. The notion of privacy that we
  employ is a variant of \emph{differential privacy}. To formally
  define differential privacy, we consider databases which are
  multisets of elements from an abstract domain $\cX$, representing
  the set of all possible data records. (In our case, the data domain
  $\cX$ can be identified with subsets of the vertex set $V$ -- it
  represents the set of all possible \emph{neighbors} that a vertex
  might be adjacent to in the network).

\begin{definition}[Differential Privacy~\citep{DMNS06}]
Two databases $D, D'\subset \cX$ are \emph{neighbors} if they differ
in at most one data record: that is, if there exists an index $i$ such
that for all indices $j \neq i$, $D_j = D'_j$.  An algorithm
$\cA \colon \cX^n \rightarrow
\cR$ satisfies $(\eps, \delta)$-differential privacy if for every
set of outcomes $S \subseteq \cR$ and for all neighboring databases
$D, D'\in
\cX^n$, the following holds:
  \[
  \Pr[\cA(D) \in S] \leq \exp(\eps) \Pr[\cA(D') \in S] + \delta.
  \]
If $\delta = 0$, we say $\cA$ satisfies $\eps$-differential privacy.
\end{definition}

This notion of privacy is very strong --- indeed, it is \emph{too}
strong for our purposes. It provides a symmetric guarantee that does
not allow the output of the algorithm to change substantially as a
function of any person's data changing. However, in our case, in order to achieve good utility guarantees we want
our algorithm to be allowed to be highly sensitive in the data of the
targeted individuals. We will modify this definition in the following
way: first, we will view the partition of the vertices into protected
and targeted individuals as a fixed, immutable characteristic,
separate from the private data of the individuals, and view the
private data of individual $v$ as being the edges incident on
$v$,\footnote{We do not provide any privacy guarantees for what we can
reveal about $v$'s membership in $\cT$. This is an inherent characteristic of the problem that we consider since the goal of our algorithms is to identify the members of $\cT$.}
\begin{equation}
\label{eq:adj}
D_v(G) = \{(u, v) \in E \mid u\in V \}.
\end{equation}
We then redefine the neighboring relation in the setting of networks:
for any protected and targeted subpopulations $\cP$ and $\cT$, two
networks $G$ and $G'$ are \emph{neighboring} if $G'$ can be obtained
by only changing a single \emph{protected} node's edges in
$G$. Specifically, $G = (V, E)$ and $G' = (V, E')$
are \emph{neighbors} with respect to a partition $V = \cP \cup \cT$ if
there exists a $v \in \cP$ such that for all $v' \neq v$:
$D_{v'}(G) \cup \{(v, v')\} = D_{v'}(G') \cup \{(v, v')\}$. Note that
for neighboring graphs $G$ and $G'$, the edge sets in the subgraph
induced on the vertices $\cT$ must also be the same.

In the following, we denote the set of all possible networks over the
vertices $V$ by $\cG$, and denote the set of all possible outcomes of an algorithm by $\cO$.



\begin{definition}[Protected Differential Privacy]
An algorithm $\cA\colon\cG \rightarrow \cO$ satisfies
$(\eps,\delta)$-\emph{protected differential privacy} if for every
partition of $n$ vertices $V$ into sets $\cP$ and $\cT$, for every pair of graphs $G, G'$ that are
neighbors with respect to the partition $(\cP, \cT)$, and for any set
of outcomes $S\subseteq
\cO$
\[
\Pr[\cA(G) \in S] \leq \exp(\eps)\Pr[\cA(G') \in S] + \delta.
\]
If $\delta = 0$, we say $\cA$ satisfies $\eps$-protected differential privacy.
\end{definition}

When the partition $(\cP, \cT)$ is clear from the context we will omit it to simplify the presentation.
In the context of graph search algorithms that we consider here the algorithm $\cA$  is given 
an oracle $\cI$ which encodes the partition $V$ into $\cP$ and $\cT$. We denote such algorithms as $\cA_\cI$. The output $\cO$ in the above definition of \textit{protected differential privacy} in the context of graph search algorithms is an ordered list of targeted individuals.

\begin{remark}
A careful reader may already have noticed that there is a trivial
  graph search algorithm that achieves $0$-protected differential
  privacy while outputting the entire set of targeted individuals
  $\cT$ --- it simply queries $\cI(v)$ for every $v \in V$, and
  outputs every $v$ such that $\cI(v) = 1$. This algorithm satisfies
  perfect (i.e. with $\epsilon = 0$) protected differential-privacy
  because it operates \emph{independently} of the private network $G$.
  The problem with this approach is that it requires querying the
  status of every vertex $v \in V$, which can be impractical both
  because of cost (the query might itself require a substantial
  investment of resources) and because of societal norms (it may not
  be defensible to subject every individual in a population to
  background checks). Hence, here we aim to design algorithms that use the
  graph data $G$ to effectively guide the search for which vertices
  $v$ to query. This is what leads to the tension with privacy, and
  our goal is to effectively trade off the privacy parameter
  $\epsilon$ with the number of queries to $\cI$ that the algorithm
  must make.
\end{remark}

One way to interpret protected differential privacy is differential privacy applied to an appropriately defined input. Let the algorithm have two
inputs: the set of edges incident to the protected vertices in $\cP$,
and the edges in $E(\cT) = \{(u, v)\mid u, v\in \cT\}$ (i.e. all of
the other edges)\footnote{It is crucial here that such a simplification can only be made for the purposes of the analysis only. Since all our algorithms are only given access to a membership oracle $\cI$ there is no way for them to explicitly construct these two inputs without incurring a cost associated with oracle queries.}. In this view, protected differential privacy only
requires the algorithm to be differentially private in its first
argument, and not in its second. This view, formalized in the
following lemma, will allow us to apply some of the basic tools of
differential privacy in order to achieve protected differential
privacy.

\begin{lemma}
\label{lem:pdp}
Given a graph $G = (V, E)$ and a partition its edges into $E_2 = E(\cT)$
and $E_1 = E \setminus E_2$ an algorithm $\cA_{\cI}(G)
:= \cA_{\cI}(E_1,E_2)$ satisfies $(\eps,\delta)$-protected
differential privacy if it is $(\eps,\delta)$-differentially private
in its first argument.
\end{lemma}

\subsection{Basic Privacy Tools}
We include some basic privacy tools here to facilitate the discussion
of our algorithm for the rest of the paper. For simplicity, we will
state these tools in the generic setting, in which we view algorithms
to be arbitrary randomized mappings from $\cX^n$ to $\cR$. 

A basic, but extremely useful result is that differential privacy is robust to arbitrary \emph{post-processing}:

\begin{lemma}[Post-Processing \cite{DMNS06}]
  For any algorithm $\cA: \cX^n \to \cR$ and any (possibly
  randomized) function $p: \cR \to \cR'$, if $\cA(\cdot)$ is $(\eps, \delta)$-differentially
  private then $p(\cA(\cdot))$ is $(\eps, \delta)$-differentially private.
\label{lem:post}
\end{lemma}

Another extremely useful property of differential privacy is that it
is compositional --- given two differentially private algorithms,
their combination remains differentially private, with parameters that
degrade gracefully. In fact, there are two such composition
theorems. The first, simpler one lets us simply add the privacy
parameters when we compose mechanisms:

\begin{lemma}[Basic Composition~\cite{DMNS06}]
  If  $M_1: \cX^n \to \cR_1$ is $(\eps_1,
  \delta_1)$-differentially private, and $M_2:\cX^n
  \times \cR_1 \to \cR_2$ is $(\eps_2,\delta_2)$-differentially private
  in its first argument, then $M: \cX^n \to \cR_2$ is
  $(\eps_1+\eps_2,\delta_1+\delta_2)$ differentially private where
$$
M(D) = M_2(D, M_1(D)).
$$
\label{lem:comp}
\end{lemma}
We can of course apply the composition theorem repeatedly, and so the
composition of $m$ mechanisms, each of which is
$\epsilon$-differentially private is $m\epsilon$-differentially
private.  The second composition theorem, due to~\cite{DRV10}, allows
us to compose $m$ mechanisms while letting the $\epsilon$ parameter
degrade \emph{sublinearly} in $m$ (at a rate of only $O(\sqrt{m})$),
at the cost of a small increase in the $\delta$ parameter.

\begin{lemma}[Advanced Composition \cite{DRV10}]
Fix $\delta>0$.  The class of $(\eps',\delta')$-differentially private
mechanisms satisfies $(\eps,m\delta' + \delta)$-differential privacy
under $m$-fold adaptive composition \footnote{See
    \cite{DRV10} for a formal exposition of adaptive composition.} for
$$
\eps ' = \frac{\eps}{\sqrt{8 m \log(1/\delta)}}.
$$
\label{lem:advanced_comp}
\end{lemma}

When designing private algorithms, we will work extensively with function 
\emph{sensitivity} for functions defined on data sets --- which informally, measures how much the
function value can change when a single data entry in the input
data set is altered.

\begin{definition}[Sensitivity]
 The sensitivity $\Delta f$ of a function $f : \cX^n \to \RR$ is
 defined as
 $$
\Delta(f) = \max_{D\sim D'} |f(D) - f(D')|,
 $$ where $D\sim D'$ indicates that $D$ and $D'$ are neighboring
databases.
\end{definition}

We will give different notions of sensitivity in the next section, which are more appropriate for some tasks in our setting. Finally, we introduce two simple algorithms that provide differential  privacy by adding noise proportional to the sensitivity of a function.

For any function $f \colon \cX^n\rightarrow \RR$, the \emph{Laplace
  mechanism} applied to function $f$ is the algorithm which on input $D$ releases $\hat f(D) := f(D) + \nu$,
where $\nu \sim \Lap(\Delta(f) / \eps)$ and $\Lap(b)$ denotes the
Laplace distribution centered at 0 with probability density function
\[
\frac{1}{2b} \exp\left(- \frac{|x|}{b}\right).
\]

\begin{lemma}[\citep{DMNS06}]\label{lem:lap}
  The Laplace mechanism is $\eps$-differentially private.
\end{lemma}

Another simple algorithm, useful for answering non-numeric queries, is
the \emph{Report Noisy Max} mechanism: given a database $D\in \cX^n$
and a collection of $k$ functions $f_1, f_2, \ldots, f_k$ each with
sensitivity at most $\gamma$, Report Noisy Max performs the following
computation:
\begin{itemize}
\item Compute the noisy estimate of each function evaluated on $D$: $\hat f_i := f_i(D)
  + \nu$ where $\nu\sim \Lap(\gamma/\eps)$;
\item Output the index $i^* = \arg\max_{i} \hat f_i$, and also the noisy value $\hat f_{i^*}$.
\end{itemize}

\begin{lemma}[\citep{DR14}]\label{lem:rnm}
  The Report Noisy Max mechanism is $2\eps$-differentially private.
\end{lemma}

\section{Statistics of Proximity (SoP)}
\label{sec:sos}
Our family of graph search algorithms will rely on various
network-centric “statistics of proximity” (SoP) that ascribe a
numerical measure of proximity of an individual vertex $v$ based on its
position in the network relative to a set $S$ of vertices from the
targeted population (which will in our usage always be the set of
targeted individuals discovered so far by the search
algorithm). Specifically, a statistic of proximity is a function $f$
that maps a network $G$, a node $v$, and a set of nodes $S\subseteq
\cT$ to a real number. Since the value $f(G, v, S)$ can reveal
information about the links in the network, we will often need to
perturb the values of these statistics with noise, calibrated with
scale proportional to the \emph{targeted sensitivity} --- the maximum
change in any \emph{targeted} node's SoP relative to any set $S$ when
a \emph{protected node}'s adjacency list is changed.

\begin{definition}[Targeted Sensitivity]
  Let $f\colon \cG \times V\times 2^{\cT} \rightarrow \RR$ be a statistic
  of proximity. The \emph{targeted sensitivity} of $f$ is
\[
\Delta(f) = \max_{G\sim G', t\in \cT, S \subseteq \cT} \left|f(G, t,
S) - f(G', t, S)\right|,
\]
where $G\sim G'$ indicates that $G$ and $G'$ are neighboring graphs
in $\cG$ relative to a fixed partition of $V$ into $\cP$ and $\cT$.
\end{definition}
Note that when computing the \emph{targeted sensitivity} we are not concerned with the effect that a change in the edges incident on vertices in $\cT$ has on the statistic, nor on the effect of any change on the statistic computed on vertices $v \in \cP$.

Another quantity of interest is \emph{impact cardinality} --- the
maximum \emph{number} of nodes whose SoP's can change as the result of a change to the
adjacency list of a single node $v \in \cP$:

\begin{definition}[Impact Cardinality]
Let $f\colon \cG \times V\times 2^{\cT} \rightarrow \RR$ be a statistic of
proximity. The \emph{impact cardinality} of $f$ is
\[
\IC(f) = \max_{G\sim G', S \subseteq \cT} \left|\{ v\in V\mid f(G, v, S) \neq
f(G', v, S)\}\right|.
\]
\end{definition}

We include some examples of candidate SoPs and their sensitivities. A desirable property for good
statistics is that they should have low sensitivity (relative to the scale of the statistic)
and small impact cardinality (relative to the target number of
queries to the identity oracle), which will allow us to achieve protected differential privacy by adding only small amounts of
noise to the various parts of our computations.

\begin{itemize}

\item $\sf{Flow}_k$$(G, v, S)$: the value of the maximum flow that can be routed between
  node $v$ and the nodes in $S$ using only paths of length at most $k$;

\item $\sf{Path}_k$$(G, v, S)$: the number of paths from $v$ to
  nodes in $S$ with length at most $k$;

\item $\sf{Triangle}$$(G, v, S) = \left|\{ \{a, b\}\subseteq S \mid
  a, b, v\text{ forms a triangle in } G\}\right|$, the number of triangles formed by the vertex $v$ in $G$;

\item $\sf{CN}$$(G, v, S) = \left|\{ u\mid (v, u)\in E \mbox{ and }
  (u, v')\in E \mbox{ for some } v'\in S\}\right|$, the number of common
  neighbors $v$ has with vertices in $S$.
\end{itemize}

In graphs with maximum degree $d$, the sensitivity of these SoPs are as follows:
\begin{itemize}
\item $\Delta({\sf{Flow}_k})\leq d$ since a vertex of degree $d$ can only affect the size of the flow by at most $d$.
\item $\Delta\left( \sf{Path}_k \right) \leq (k - 1) d^{k-1}$  since the total number of paths from $v$ to $S$ on which a vertex $u \in \cP$ might lie is at most $\sum_{j = 1}^{k - 1} d^{j - 1} d^{k - j} = (k - 1) d^{k - 1}$. Here we used the index $j$ to denote the index of $u$ along the path starting from $v$ together with the fact that the total number of different paths of length $\ell$ from $u$ is at most $d^{\ell}$.
\item $\Delta\left( \sf{Triangle}\right) \leq d$ since each triangle is associated with an edge and the total number of edges affected is at most $d$.
\item $\Delta\left( \sf{CN} \right) \leq 1$ since a single vertex can change the count of common neighbors by at most 1.
\end{itemize}
Note that ${\sf{Path}}_1(G,v,S)$, which simply counts the number of
edges between $v$ and $S \subseteq \cT$ actually has targeted
sensitivity \emph{zero}. This is because, since $S \subseteq \cT$, if
$v \in \cT$ is also a member of the targeted population, then the
statistic is a function only of $E(\cT)$, the edge set of the subgraph
defined over the targeted sub-population $\cT$. Since $E(\cT)$ is
identical on all neighboring graphs, and because targeted sensitivity
only measures the sensitivity of the SoP evaluated on \emph{targeted}
nodes to changes in \emph{protected} nodes, we get zero
sensitivity. This will be important to our analysis.

\section{SoP Based Targeting Algorithms}
Before we present the full algorithm, we will first present some
useful subroutines together with analysis of their privacy properties.

\subsection{Statistic-First Search }
First, we introduce statistic-first search (\bfs), a search algorithm
that explores the entire targeted connected component given a seed
targeted node. It is a search strategy that only inspects the
neighbors of verified targeted nodes. The formal description is
presented in~\Cref{alg:bfs}.
\begin{algorithm}[h]
\caption{$\bfs(G, t)$}
\label{alg:bfs}
\begin{algorithmic}[0]
\STATE {\bf Input}: known targeted $t$ in a network $G$
\STATE {\bf Initialize}:
\[
\tilde T = \{t\} \qquad I = \{t\} \qquad \cN = \text{ neighbors of } t
\]
\STATE {\bf while} $\cN \setminus I \neq \emptyset$
\INDSTATE {Let \[v' = \argmax_{x \in \cN \setminus I}{\sf{Path}}_1(G, x, \tilde T)\]}
\INDSTATE {Query $\cI(v')$ to determine $v'$'s targeted status.}
\INDSTATE {$I = I\cup \{v'\}$}

\INDSTATE{\textbf{if} $\cI(v') = 1$ \textbf{ then } $\tilde T = \tilde T
  \cup \{v'\} \mbox{ and } \cN = \mbox{ neighbors of } \tilde T $}

\STATE {\bf Output}: list $\tilde T$
\end{algorithmic}
\end{algorithm}

We now establish the simple but remarkable privacy guarantee of $\bfs$ --- the algorithm
can often identify a targeted connected component free of privacy cost.
\begin{lemma}
\label{lem:bfsp}
The graph search algorithm $\bfs$ satisfies $0$-protected differential
privacy.
\end{lemma}

\begin{proof}
  Let $G$ and $G'$ be two neighboring networks in $\cG$ with respect to the same partition $(\cP, \cT)$. We know that
  both networks have the same set of targeted nodes $\cT$ and targeted
  links $E(\cT)$. Since we know that $\Delta({\sf Path_1}) = 0$, and
  $\bfs$ only branches on the evaluations of $f$ on nodes $v \in \cT$,
  the behavior of $\bfs$ depends only on $\cT$ and $E(\cT)$,
  and hence $\bfs(G, v, f)$ and $\bfs(G', v, f)$ always produce the
  same output.
\end{proof}

\subsection{Private Search for Targeted Component}
With $\bfs$, we can start with a seed node $v \in \cT$, and at no additional privacy cost, find the entire connected component $T \subseteq \cT$ in the subgraph defined on vertices in $\cT$ that $v$ belongs to. This is a useful subroutine in a graph search algorithm: however, once we have exhausted our seed node's connected component $T$, we need a way to search for a new seed node $v'$ that is part of a new connected component. This is what our subroutine $\searc$ does:

\begin{enumerate}
\item Given a list of already identified members of the targeted subpopulation $\tilde T \subseteq \cT$ and a SoP, we compute a noisy SoP value for each node $\hat f(v)$;
\item we sort the nodes in decreasing order of their noisy SoP value, and query each vertex $v$ in this order to determine whether $v \in \cT$ or $v \in \cP$ until we find a node such that $v \in \cT$.
\item If we query $K$ nodes without having found any members of the targeted sub-population, we halt the search. The stopping condition needs to be checked privately, so $K$ is in fact a randomly perturbed value.
\end{enumerate}
We include a formal description of $\searc$ in~\Cref{alg:sc}.

\begin{algorithm}[h]
\caption{$\searc(G, \tilde T, I, f, \eps, K)$}
\label{alg:sc}
\begin{algorithmic}[0]
\STATE {\bf Input}: identified members of the targeted population $\tilde T \subseteq \cT$ in a network $G$,
the set of investigated nodes $I$, SoP $f$, privacy parameter $\eps$,
and stopping threshold $K$

\STATE {\bf Initialize}:
\[
\mbox{Set noisy stopping threshold }\hat K = K + \nu \mbox{ where } \nu \sim
\Lap(2\, \IC(f) / \eps)\mbox{ and } \mbox{count}=0
\]
\STATE {\bf for} each $v \in V \setminus I$:

\INDSTATE {let } $\hat f(v) = f(G, v, \tilde T)+ \zeta_v$ where
$\zeta_v\sim \Lap(4\Delta(f)/\eps)$

\STATE {\textbf{while} $(V\setminus I) \neq\emptyset$ and count$\leq \hat K$}
\INDSTATE {Let \[v' = \argmax_{x \in V \setminus I}\hat f(v)\]}
\INDSTATE {Query $\cI(v')$ to determine if $v' \in \cT$.}
\INDSTATE {Let $I = I\cup \{v'\}$ and count $ = $ count $+1$}

\INDSTATE{\textbf{if}  $\cI(v') = 1$  \textbf{ then return } $\{v'\}$}

\STATE{\bf{return} $\emptyset$}
\end{algorithmic}
\end{algorithm}

\begin{lemma}
\label{lem:searchp}
The targeting algorithm $\searc$ instantiated with privacy parameter
$\eps$ satisfies $\eps$-protected differential privacy.
\end{lemma}

\begin{proof}
Let $G$ and $G'$ be neighboring networks in $\cG$. First suppose that $\tilde T = \cT$. In this case, $\searc$ will output $\emptyset$ with probability $1$ on both inputs $G$ and $G'$, and hence satisfy $0$-protected differential privacy. Hence, for the remainder of the argument, we can assume that there exists a vertex $v \in \cT \setminus \tilde T$.

In this case, the algorithm can equivalently be viewed as the following 2-step procedure:
\begin{enumerate}
\item Use the Report Noisy Max algorithm to output the index of the targeted node
  $t$ which maximizes $\hat f(t)$ \emph{together with the perturbed value} $\hat f(t)$;
\item Let $b$ be the number of nodes $i \in \cP$ such that $\hat f(i)
  \geq \hat f(t)$. If $b - \nu \leq K$, output node $t$. Otherwise output $\emptyset$.
\end{enumerate}

When the input to the algorithm $G = (V,E)$ is viewed as the pair of edge sets $(E \setminus E(\cT), E(\cT))$, we show below that each of these two steps satisfies $\eps/2$-differential
privacy with respect to its first argument. By the basic composition
theorem~\Cref{lem:comp} and~\Cref{lem:pdp} the algorithm $\searc$
satisfies $\eps$-protected differential privacy.

The first step is an instantiation of the Report Noisy Max algorithm
with privacy parameter $\eps/2$, so it is $\eps/2$-differentially
private by~\Cref{lem:rnm}.

The second step is post-processing of $\hat b = b - \nu$. We just need
to show that releasing $\hat b$ is $\epsilon/2$ differentially
private, and the result will follow by the post processing
lemma:~\Cref{lem:post}. For the following analysis, we will fix the
(arbitrary) values of $\{\zeta_v\}$ and compute probabilities as a
function of the randomness of $\nu$.

Since the SoP $f$ has impact cardinality $\IC(f)$, we know that there are
at most $\IC(f)$ many nodes $v$ such that $f(G, v, \tilde T) \neq f(G'
, v, \tilde T)$. It follows that
\[
\left|\{ v \mid f(G, v, \tilde T) + \zeta_v \neq f(G', v, \tilde T) + \zeta_v\}\right| \leq \IC(f)
\]
Let $b(G)$ and $b(G')$ denote the number of nodes $i \in \cP$ with
$\hat f(v) \geq \hat f(t)$ in $G$ and $G'$ respectively. Then we know
that,
\[
|b(G) - b(G')|\leq \IC(f).
\]
Since we are releasing $\hat b$ by adding noise $\nu$ sampled from the
Laplace distribution $\Lap(2\,\IC(f)/\eps)$, $\hat b$ is
$\eps/2$-differentially private from the property of Laplace
mechanism~\Cref{lem:lap}.
\end{proof}

\subsection{The Full Algorithm: Putting the Pieces Together}

Our graph search algorithm alternates between two phases. In the first
phase, the algorithm starts with a seed node $v \in \cT$, and uses
$\bfs$ to find every other vertex $v' \in \cT$ that is part of the
same connected component as $v$ in the subgraph defined on
$\cT$. After this targeted component has been fully identified, the
second phase begins. In the second phase, the algorithm uses $\searc$
to search for a new vertex $v \in \cT$ that will serve as a seed node
for the next iteration of $\bfs$. Once such a seed node has been
found, the algorithm reverts to phase $1$, and this continues for a
specified number of iterations.  The formal description of the
algorithm in presented in~\Cref{alg:private}.

\begin{algorithm}[h]
\caption{Private Search Algorithm: $\pmc(G, \cS, f, k, N, \eps)$}
\label{alg:private}
\begin{algorithmic}[0]
  \STATE {\bf Input}: A network $G$, a seed node $\cS \in \cT$, a SoP $f$ for $\searc$, a target number
  of components $k$, a stopping threshold $N$ for
  $\searc$, and privacy parameter $\eps$

  \STATE{\bf Initialize}:
   Use set $I$ to keep track of the set of investigated nodes. Initially $I = \{\cS\}$.

  \STATE{\bf Let}  list  $\tilde T = \bfs(G, \cS)$
  \STATE{\bf For}
  $k-1$ rounds:
  \INDSTATE{\bf let} $a = \searc(G, \tilde T, I, f, \eps, N)$
  \INDSTATE{\bf if} $a = \emptyset$
  \INDSTATE[2]{\bf Output } $\tilde T$
  \INDSTATE{\bf else}
  \INDSTATE[2]{\bf let $\tilde T = \tilde T \cup \bfs(G, a)$}
  \STATE{\bf Output} $\tilde T$
\end{algorithmic}
\end{algorithm}

We now establish the following privacy guarantee
of~\Cref{alg:private}. Recall that the parameter $k$ represents the
maximum number of disjoint components of the subgraph defined on $\cT$
that the algorithm will identify.
\begin{theorem}
Fix any $0< \delta <1$.  $\pmc(\cdot, \cdot,
\cdot, \cdot, k, \cdot, \eps)$ satisfies $\eps_1$-protected
differential privacy for
\[
\eps_1 = (k - 1)\eps,
\]
and satisfies $(\eps_2, \delta)$-protected differential privacy for
\[
\eps_2 = 2\sqrt{2(k-1) \ln(1/\delta)} \eps.
\]

\end{theorem}

\begin{proof}
The algorithm is a composition of at most $k$ instantiations of $\bfs$
and $(k-1)$ instantiations of $\searc$ with privacy parameter
$\eps$. Recall that each call to $\bfs$ is 0-differentially private,
and each call to $\searc$ is $\eps$-differentially private with respect to the edges incident on vertices in $\cP$. By the composition theorem, we know that
the algorithm is $(k-1)\eps$-differentially private
by~\Cref{lem:comp}, and at the same time $(\sqrt{8 k \ln(1/\delta)},
\delta)$-differentially private for any $\delta\in (0, 1)$
by~\Cref{lem:advanced_comp}. Our result then easily follows
from~\Cref{lem:pdp}.

\end{proof}

\section{Experiments}

\subsection{Subpopulation Construction}
Our experiments are conducted on two real social networks:

\begin{itemize}
\item the scientific collaboration network in DBLP (``Digital
  Bibliography and Library Project''), where nodes represent authors
  and edges connect authors that have coauthored a paper;
\item the movie costarring network in IMDB  (``Internet Movie Database''), where nodes represent
  actors and edges connect actors that have appeared in a movie
  together.
\end{itemize}

\paragraph{Pre-processing step on the networks:}{We
 sparsify the IMDB and DBLP networks by removing a subset of the
 edges. This will allow us to generate multiple targeted components
 more easily. In both networks, there is a natural notion of weights
 for the edges. In the case of DBLP, the edge weights correspond to
 the number of papers the individuals have co-written. In the case of
 IMDB, the edge weights correspond to the number of movies two actors
 have co-starred in. In our experiments, we only remove edges with
 weights less than 2.  }

However, the networks we use do not have an identified partition of
the vertices into a targeted and protected subpopulation. Instead, we
generate the targeted subpopulation synthetically using the following
diffusion process. We use the language of ``infection'', which is
natural, but we emphasize that this process is not specific to our
motivating example of the targeted population representing people
infected with a dangerous disease. The goal of the infection process
is to generate a targeted subpopulation $\cT$ such that:
\begin{enumerate}
\item The subnetwork restricted to $\cT$ has multiple distinct
  connected components (so that the search problem is algorithmically
  challenging, and isn't solved by a single run of statistic-first
  search), and
\item The connected components of $\cT$ are close to one another in
  the underlying network $G$, so that the network data is useful in
  identifying new members of $\cT$.
\end{enumerate}
 
 The process $\infect(G, s, p, q, k)$ takes as input a seed infected
 node $s$, two values $p, q\in (0, 1)$, and a number of rounds $k$,
 and proceeds with two phases: 
\begin{enumerate}
  \item \textbf{Infection phase:} Initially, only the node $s$ is in
    the infected set $\tilde\cI$. Then in each of the $k$ rounds, each
    neighbor $v$ of the infected nodes $\tilde\cI$ becomes infected
    independently with probability $p$.
  \item \textbf{Immune phase:} After the infection process above, we
    will set some of the infected nodes as immune. For each node $i$
    in the infected node set $\tilde \cI$, let $i$ become ``immune''
    (non-infected) with probability $q$. 
\end{enumerate}

We include a formal description of the algorithm in~\Cref{alg:infect}.

\begin{algorithm}[h]
\caption{$\infect(G, s, p, q, k)$}
\label{alg:infect}
\begin{algorithmic}[0]
\STATE {\bf Input}: a network $G$, a seed node $s$ in $G$, infection
probability $p$, and immune probability $q$
\STATE {Initially the infected population contains only the seed node}:
\[
\tilde I = \{s\} 
\]
\STATE {\bf for} $t = 1, \ldots, k$:
\INDSTATE {{\bf for} each node $v$ that is neighbor to $\tilde I$:}
\INDSTATE[2] Let $\nu$ be a uniformly random number from $[0, 1]$
\INDSTATE[2] {\bf if} $\nu \leq p$ {\bf then} $\tilde I = \tilde I \cup \{v\}$

\STATE{{\bf let} $\cT = \emptyset$}
\STATE {\bf for} each node $v'\in \tilde I$:
\INDSTATE let $\nu$ be a uniformly random number from $[0,1]$
\INDSTATE {\bf if} $\nu > q$ {\bf then} $\cT = \cT \cup \{v\}$ 

\STATE {\bf Output}:  $\cT$ as the targeted subpopulation
\end{algorithmic}
\end{algorithm}

\subsection{Non-Private Benchmark $\mc$}
We experimentally evaluate the performance of our
algorithm~\Cref{alg:private} on two social network data-sets with a
partition of vertices into $\cP$ and $\cT$ using the infection process
described in the previous section. We compare the performance of the
private version of our algorithm with the non-private
version~\Cref{alg:template} which uses the SoP directly, without
adding noise. The metric we are interested in is how many queries to
the identify oracle $\cI$ are needed by each algorithm to find a given
number of members of the protected sub-population $\cT$. We here give
a formal description of the non-private version of our graph search
algorithm in~\Cref{alg:template}.

\begin{algorithm}[h]
\caption{Non-Private Targeting Algorithm: $\mc(G, \cS, f, k, N)$}
\label{alg:template}
\begin{algorithmic}[0]
  \STATE {\bf Input}: A network $G$, a seed node $\cS \in \cT$, a SoP $f$ for $\searc$, a target number of components to find $k$, and a stopping threshold $N$ for $\searc$
  \STATE{\bf Initialize}:
   Use $I$ to keep track of the set of investigated nodes. Initially $I = \{\cS\}$.
  \STATE{\bf Let}  list $\tilde T = \bfs(G, \cS)$
  \STATE{\bf For} $k-1$ rounds:
\INDSTATE {\bf let} $\text{count} = 0$ and $a = \emptyset$
\INDSTATE {\textbf{while} $(V\setminus I) \neq\emptyset$ and count$\leq N$}
\INDSTATE[2] {Let \[v' = \argmax_{x \in V \setminus I} f_2(G, v, \tilde T)\]}
\INDSTATE[2] {Query $\cI(v')$ to determine if $v' \in \cT$.}
\INDSTATE[2] {Let $I = I\cup \{v'\}$ and count $ = $ count $+1$}
\INDSTATE[2]{\textbf{if}  $\cI(v')= 1$ \textbf{ then let } $a = v'$ and \textbf{break}}
  \INDSTATE{\bf if} $a = \emptyset$ \textbf{ then Output } $\tilde T$
  \INDSTATE{\bf else let}  $\tilde T = \tilde T \cup \bfs(G, a ,f_1)$
  \STATE{\bf Output} $\tilde T$
\end{algorithmic}
\end{algorithm}

\subsection{SoP Instantiations}

In our experiments, we will use the SoP 
$\CN$ for the $\searc$ subroutine, which is the number of common
neighbors between the node $v$ and the subset of nodes $S$
representing the already discovered members of the targeted
population. The targeted sensitivity of $\CN$ is bounded by 1.
\begin{lemma}
  The SoP $\CN$ has targeted sensitivity $\triangle(\CN)$ bounded by
  1.
\end{lemma}

\begin{proof}
  Let $G$ and $G'$ be two neighboring networks over the same protected
  and targeted populations $\cP$ and $\cT$. Let $t\in \cT$ be a
  targeted node and $S\subseteq V$ be a subset of nodes. Since $G$ and
  $G'$ only differ by the edges associated with a protected node $i$,
  we know that the neighbor sets of both $t$ and $S$ can differ by at
  most one node between $G$ and $G'$. Note that the $\CN(G, t, S)$
  computes the cardinality of the intersection between these two sets,
  and the intersection sets of these two networks can differ by at
  most one node. It follows that $\triangle(\CN) \leq 1$.
\end{proof}





\end{document}